\newif\ifstoc
\stoctrue 
\stocfalse

\newcommand{\stocoption}[2]
{\ifstoc%
#1%
\else%
#2%
\fi}

\ifstoc
\documentclass{llncs}
\else
\documentclass[11pt]{article}

\usepackage{typearea}
\typearea{14}

\usepackage[compact]{titlesec}

\fi

\usepackage{epsfig}
\usepackage{amsfonts}
\usepackage{amssymb}
\usepackage{amstext}
\usepackage{amsmath}
\usepackage{xspace}
\usepackage{algorithm}
\usepackage[noend]{algorithmic}

\allowdisplaybreaks
\allowdisplaybreaks

\ifstoc

\newtheorem{cl}[theorem]{Claim}
\newtheorem{fact}{Fact}
\else

\newtheorem{theorem}{Theorem}[section]
\newtheorem{definition}[theorem]{Definition}

\newtheorem{lemma}[theorem]{Lemma}
\newtheorem{fact}[theorem]{Fact}
\newtheorem{cl}[theorem]{Claim}

\newenvironment{proof}{{\bf Proof:  }}{\hfill\rule{2mm}{2mm}}

\newcommand{\ts}{\textstyle}

\fi

\numberwithin{algorithm}{section}

\allowdisplaybreaks

\newenvironment{proofclaim}{

\noindent{\bf Proof:}}
{\hfill$\blacktriangleleft$

}

\newenvironment{proofof}[1]{

\noindent{\bf Proof of {#1}:}}
{\hfill$\blacksquare$

}

\newcommand{\junk}[1]{}
\newcommand{\ignore}[1]{}

\newcommand{\F}[0]{{\ensuremath{\mathcal{F}}}}
\newcommand{\G}[0]{{\ensuremath{\mathcal{G}}}}

\def\abs#1{\mathopen| #1 \mathclose|}   

\def\f {\ensuremath{\mathcal{F}}\xspace}

\def\tf{\mathcal{T}}

\def\g{\mathcal{G}}

\def\f{\ensuremath {\mathcal{F}}\xspace}

\def\m{\mathcal{M}}
\def\s{\mathcal{S}}

\def\opt{{\sf Opt}\xspace}

\def\cov{{\sf Cov}\xspace}
\def\rcov{{\sf RobCov}\xspace}

\def\phis{\phi^*}
\def\ms{{\sf MinSet}}
\def\lp{{\sf LP}}

\newcommand{\argmin}{\operatorname{argmin}}

\newcommand{\sse}{\subseteq}

\newcommand{\A}{{\mathbb{A}}}
\newcommand{\hatA}{{\widehat{A}}}

\newcommand{\net}{N}

\newcommand{\initOneLiners}{%
    \setlength{\itemsep}{0pt}
    \setlength{\parsep }{0pt}
    \setlength{\topsep }{0pt}
}
\newenvironment{OneLiners}[1][\ensuremath{\bullet}]
    {\begin{list}
        {#1}
        {\initOneLiners}}
    {\end{list}}

\newcommand{\Phistar}{\ensuremath{\Phi^*}\xspace}
\newcommand{\theeta}{W}
\newcommand{\etchj}{{H'}}
\newcommand{\cut}{{X}}

\title{Thrifty Algorithms for Multistage Robust Optimization}

\stocoption{
\author{ Anupam Gupta\inst{1}  \and Viswanath Nagarajan\inst{2} \and Vijay V.~Vazirani\inst{3} }

\institute{Computer Science Department, Carnegie Mellon University. \and IBM T.J. Watson Research Center. \and College of Computing, Georgia Institute of Technology.}

\date{}
}
{
\author{
Anupam Gupta\thanks{Computer Science Department, Carnegie Mellon
    University, Pittsburgh, PA 15213, USA. Supported in part by
    NSF award CCF-1016799  and an Alfred P.~Sloan
    Fellowship. Email: anupamg@cs.cmu.edu}
\and
Viswanath Nagarajan\thanks{IBM T.J. Watson Research Center, Yorktown
  Heights, NY 10598, USA. Email: viswanath@us.ibm.com}
\and
Vijay V.~Vazirani\thanks{College of Computing, Georgia Institute of Technology, Atlanta, GA 30332-0280. Supported by NSF Grants CCF-
0728640 and CCF-0914732, ONR Grant N000140910755, and a Google Research
Grant. Email: vazirani@cc.gatech.edu}
}
}

\begin{document}
\maketitle
\begin{abstract}
We consider a class of multi-stage robust covering problems, where additional   information is revealed about the problem instance in each stage, but
  the cost of taking actions increases. The dilemma for the
  decision-maker is whether to wait for additional information and risk
  the inflation, or to take early actions to hedge against rising
  costs. We study the ``$k$-robust'' uncertainty model: in each stage 
$i = 0, 1, \ldots,  T$, the algorithm is shown some subset of size $k_i$ that completely  contains the eventual demands to be covered; here $k_1 > k_2 > \cdots > k_T$ which ensures increasing information over time. The goal is to minimize the cost incurred in the
  \emph{worst-case} possible sequence of revelations.

  \smallskip For the {\em multistage $k$-robust set cover} problem, we give an $O(\log m+\log n)$-approximation algorithm, nearly matching the $\Omega\left(\log n+\frac{\log m}{\log\log m}\right)$ hardness of approximation~\cite{FJMM07} even for $T=2$ stages. Moreover, our algorithm has a useful ``thrifty'' property: it takes actions on just two stages. We show similar thrifty algorithms for multi-stage $k$-robust {\em Steiner tree}, {\em Steiner forest}, and {\em minimum-cut}. For these problems our approximation guarantees are $O(\min\{ T, \log n,
  \log \lambda_{\max} \})$, where $\lambda_{\max}$ is the maximum
  inflation over all the stages. We conjecture that these problems also
  admit $O(1)$-approximate thrifty algorithms.
\end{abstract}

\section{Introduction}
\label{sec:introduction}

This paper considers approximation algorithms for a set of multi-stage decision problems. Here, additional
information is revealed about the problem instance in each stage, but the cost of taking actions increases. The
decision-making algorithm has to decide whether to wait for additional information and risk the rising costs, or to
take actions early to hedge against inflation. We consider the model of robust optimization, where we are told what the
set of possible information revelations are, and want to minimize the cost incurred in the \emph{worst-case} possible
sequence of revelations.

For instance, consider the following multi-stage set cover problem: initially we are given a set system $(U, \F)$. Our
eventual goal is to cover some subset $A \sse U$ of this universe, but we don't know this ``scenario'' $A$ up-front.
All we know is that $A$ can be any subset of $U$ of size at most $k$. Moreover we know that on each day $i$, we will be
shown some set $A_i$ of size $k_i$, such that $A_i$ contains the scenario $A$---these numbers $k_i$ decrease over time,
so that we have more information as time progresses, until $\cap_{i=0}^T A_i = A$. We can pick sets from $\F$ toward
covering $A$ whenever we want, but the costs of sets increase over time (in a specified fashion). Eventually, the sets
we pick must cover the final subset $A$.  We want to minimize the worst-case cost
\begin{equation}
  \label{eq:1}
  \max_{\sigma = \langle A_1, A_2, \ldots, A_T \rangle : |A_t| = k_t \; \forall t }
  \text{ total cost of algorithm on sequence $\sigma$ }
\end{equation}
This is a basic model for multistage robust optimization and requires minimal specification of the  uncertainty sets (it only needs the cardinality bounds $k_i$s).

Robust versions of Steiner tree/forest, minimum cut, and other covering
problems are similarly defined. This
tension between waiting for information vs.\ the temptation to buy early
and beat rising costs arises even in 
2-stage decision problems---here we have $T$ stages of decision-making, making this more acute.

A comment on the kind of partial information we are modeling: in our setting we are given progressively more
information about events that \emph{will not} happen, and are implicitly encouraged (by the rising prices) to plan
prudently for the (up to $k$) events that will indeed happen. For example, consider a farmer who has a collection of
$n$ possible bad events (``high seed prices in the growing season'', $\{$``no rains by month $i$''$\}_{i = 1}^{5}$,
etc.), and who is trying to guard against up to $k$ of these bad events happening at the end of the planning horizon.
Think of $k$ capturing how risk-averse he is; the higher the $k$, the more events he wants to cover. He can take
actions to guard against these bad events (store seed for planting, install irrigation systems, take out insurance,
etc.). In this case, it is natural that the information he gets is about the bad events that do not happen.

This should be contrasted with online algorithms, where we are only given events that do happen---namely, demands that
need to be immediately and irrevocably covered. This difference means that we cannot directly use the ideas from online
competitive analysis, and consequently our techniques are fairly different.\footnote{It would be
  interesting to consider a model where a combination of positive and
  negative information is given, i.e., a mix of robust and online
  algorithms. We leave such extensions as future work.}  A second
difference from online competitive analysis is, of course, in the objective function: we guarantee that the cost
incurred on the worst-case sequence of revelations is approximately minimized, as opposed to being competitive to the
best series of actions for every set of revelations---indeed, the rising prices make it impossible to obtain a
guarantee of the latter form in our settings.

\smallskip
{\bf Our Results.} In this paper, we give the first approximation algorithms for standard covering problems (set cover, Steiner tree and
forest, and min-cut) in the model of multi-stage robust optimization with recourse.  A feature of our algorithms that
make them particularly desirable is that they are ``thrifty'': \emph{they actually take actions in just two stages},
regardless of the number of stages $T$. Hence, even if $T$ is polynomially large, our algorithms remain efficient and simple (note that the optimal decision tree has
potentially exponential size even for constant $T$). For example, the set cover algorithm covers
some set of ``dangerous'' elements right off the bat (on day $0$), then it waits until a critical day $t^*$ when it
covers all the elements that can concievably still like in the final set $A$. We show that this set-cover algorithm is
an $O(\log m + \log n)$-approximation, which almost matches the hardness result of $\Omega(\log n + \frac{\log m}{\log
\log m})$~\cite{FJMM07} for $T=2$.

We also give thrifty algorithms for three other covering problems: Steiner tree, Steiner forest, Min-cut---again, these
algorithms are easy to describe and to implement, and have the same structure:
\begin{quote}
  We find a solution in which decisions need to be made \emph{only at
    two points in time}: we cover a set of dangerous elements in
  stage~$0$ (before any additional information is received), and then we
  cover all surviving elements at stage $t^*$, where $t^* = \argmin_t
  \lambda_t k_t$.
\end{quote}
For these problems, the approximation guarantee we can currently prove
is no longer a constant, but depends on the number of stages:
specifically, the dependence is $O(\min \{T, \log n, \log
\lambda_{\max}\})$, where $\lambda_{\max}$ is the maximum inflation
factor. While we conjecture this can be improved to a constant, we would
like to emphasize that even for $T$ being a constant more than two, previous
results and techniques do not imply the existence of a constant-factor approximation
algorithm, let alone the existence of a thrifty algorithm.

The definition of ``dangerous'' in the above algorithm is, of course, problem dependent: e.g., for set cover these are elements which cost
more than $\opt/k_{t^*}$ to cover. In general, this defintion is such that bounding the cost of the elements we cover
on day $t^*$ is immediate. And what about the cost we incur on day $0$? This forms the technical heart of the proofs,
which proceeds by a careful backwards induction over the stages, bounding the cost incurred in covering the dangerous
elements that are still uncovered by $\opt$ after $j$ stages. These proofs exploit some net-type properties of the
respective covering problems, and extend the results in Gupta et al.~\cite{GNR10-robust}. While our algorithms appear
similar to those in~\cite{GNR10-robust}, the proofs require new technical ideas such as the use of non-uniform
thresholds in defining ``nets'' and proving properties about them.

The fact that these multistage problems have near-optimal strategies with this simple structure is quite surprising.
One can show that the optimal solution may require decision-making at all stages (we show an example for set cover in \stocoption{ the full version }{ Section~\ref{sec:lowerbound}}). It would be interesting to understand this phenomenon further. For problems other than
set cover (i.e., those with a performance guarantee depending on $T$), can we improve the guarantees further, and/or
show a tradeoff between the approximation guarantee and the number of stages we act in? These remain interesting
directions for future research.

We also observe in \stocoption{ the full version }{ Section~\ref{app:subset-k-rob}} that thrifty algorithms perform poorly for multistage robust
set-cover even on slight generalizations of the above ``$k$-robust uncertainty sets''. In this setting it turns out
that any reasonable near-optimal solution must act on all stages. This suggests that the $k$-robust uncertainty sets
studied in this paper are crucial to obtaining good thrifty algorithms.

\smallskip
{\bf Related Work.}
Demand-robust optimization has long been studied in the operations research literature, see eg. the survey article by
Bertsimas et al.~\cite{BBC-survey} and references therein.
The multistage robust model was  studied in Ben-Tal et al.~\cite{BGGN04}. Most of these works involve only continuous
decision variables. On the other hand, the problems considered in this paper involve making discrete decisions.

Approximation algorithms for robust optimization are of more recent
vintage: all these algorithms are for two-stage optimization with
discrete decision variables. Dhamdhere et al.~\cite{DGRS05} studied
two-stage versions when the scenarios were {\em explicitly} listed, and gave
constant-factor approximations for Steiner tree and facility location,
and logarithmic approximations to mincut/multicut problems. Golovin et
al.~\cite{GGR06} gave $O(1)$-approximations to robust mincut and 
shortest-paths. Feige et al.~\cite{FJMM07} considered {\em implicitly}
specified scenarios and introduced the $k$-robust uncertainty model
(``scenarios are all subsets of size $k$''); they gave an 
$O(\log m \log n)$-approximation algorithm for 2-stage $k$-robust set cover using an LP-based approach. Khandekar et al.~\cite{KKMS08} gave $O(1)$-approximations for 2-stage $k$-robust
Steiner tree, Steiner forest on trees and facility location, using a combinatorial algorithm. Gupta et
al.~\cite{GNR10-robust} gave a general framework for 
two-stage $k$-robust problems, and used it to get better results for set
cover, Steiner tree and forest, mincut and multicut. We build substantially
on the ideas from~\cite{GNR10-robust}.

Approximation algorithms for multistage stochastic optimization have been given in~\cite{SwamyS05,GuptaPRS05}; in the
stochastic world, we are given a probability distribution over sequences, and consider the average cost instead of the
worst-case cost in~(\ref{eq:1}). However these algorithms currently only work for a constant number of stages, mainly
due to the explosion in the number of potential scenarios. The current paper raises the  possibility that for
``simple'' probability distributions, the techniques developed here may extend to stochastic optimization.

\smallskip
{\bf Notation.} We use $[T]$ to denote $\{0,\cdots,T\}$, and $\binom{X}{k}$ to denote the collection of all $k$-subsets of the set $X$.

\section{Multistage Robust Set Cover}
\label{sec:set-cover}

In this section, we give an algorithm for multistage robust set cover with approximation ratio
$O(\log m + \log n)$; this approximation matches the previous best approximation guarantee for two-stage robust set
cover~\cite{GNR10-robust}. Moreover, our algorithm has the advantage of picking sets only in two stages. \stocoption{ }{ (In Section~\ref{sec:lowerbound}, we show that an optimal strategy might need to pick sets in all stages.)}

The multistage robust set cover problem is specified by a set-system $(U, \F)$ with $|U| = n$, set costs
$c:\f\to \mathbb{R}_+$, a time horizon $T$, integer values $ n = k_0 \ge k_1\ge k_2 \ge \cdots \ge k_T$, and inflation
parameters $1 = \lambda_0 \le \lambda_1\le \lambda_2 \le \cdots \le \lambda_T$. Define $A_0 = U$, and $k_0 = |U|$. A
{\em
  scenario-sequence} $\A = (A_0, A_1, A_2, \ldots, A_T)$ is a sequence
of $T+1$ `scenarios' such that $|A_i|=k_i$ for each $i\in [T]$. Here $A_i$ is the information revealed to the algorithm
on day $i$.
The elements in $\cap_{i \leq j} A_i$ are referred to as being \emph{active} on day~$j$. 
\begin{OneLiners}
\item On day $0$, all elements are deemed active and any set $S \in \f$
  may be chosen at the cost $c(S)$.
\item On each day $j\ge 1$, the set $A_j$ with $k_j$ elements is
  revealed to the algorithm, and the active elements are $\cap_{i \leq j} A_i$.  The algorithm can now pick any sets, where the
  cost of picking set $S \in \f$ is $\lambda_j \cdot c(S)$.
\end{OneLiners}
Feasibility requires that all the sets picked over all days $j \in [T]$ cover $\cap_{i \leq T} A_i$, the elements that
are still active at the end. The goal is to minimize the worst-case cost incurred by the algorithm, the worst-case
taken over all possible scenario sequences. Let $\opt$ be this worst-case cost for the best possible algorithm; we will
formalize this soon. The main theorem of this section is the following:
\begin{theorem}
  \label{thm:mult-sc}
  There is an $O(\log m +  \log n)$-approximation algorithm for the $T$-stage $k$-robust set cover problem.
\end{theorem}

The algorithm is easy to state: For any element $e\in U$, let $\ms(e)$ denote the minimum cost set in \f that contains
$e$. Define $\tau := \beta \cdot \max_{j\in [T]} \frac{\opt}{\lambda_j\,k_j}$ where $\beta:=36\,\ln m$ is some
parameter.  Let $j^* = \mbox{argmin}_{j\in
  [T]} (\lambda_j\,k_j)$. Define the ``net'' $\net:= \left\{e\in U \mid
  c(\ms(e))\ge \tau\right\}$.  Our algorithm's strategy is the
following:
\begin{quote}
  On day zero, choose sets $\phi_0 :=$ Greedy-Set-Cover($\net$).

  On day $j^*$, for any yet-uncovered elements $e$ in $A_{j^*}$, \\
  $~~~~~~~~~$ pick a
  min-cost set in $\f$ covering $e$.

  On all other days, do nothing.

\end{quote}
It is clear that this is a feasible strategy; indeed, all elements that are still active on day $j^*$ are covered on
that day. (In fact, it would have sufficed to just cover all the elements in $\cap_{i \leq j^*} A_i$.)  Note that this
strategy pays nothing on days other than $0$ and $j^*$; we now bound the cost incurred on these two days.

\begin{cl}
  \label{clm:dayjstar}
  For any scenario-sequence $\A$, the cost on day $j^*$ is at most
  $\beta\cdot \opt$.
\end{cl}

\begin{proof}
  The sets chosen in day $j^*$ on sequence $\A$ are $\{\ms(e) \mid e\in
  A_{j^*}\setminus \net\}$, which costs us
  \[ \ts \lambda_{j^*} \, \sum_{e\in A_{j^*}\setminus \net} c(\ms(e))\le
  \lambda_{j^*} \, |A_{j^*}|\cdot \tau = \lambda_{j^*} \, k_{j^*} \tau =
  \beta\cdot \opt. \] The first inequality is by the choice of $\net$,
  the last equality is by $\tau$'s definition. \stocoption{\qed}{}
\end{proof}

\begin{lemma}
  \label{th:dayzero}
  The cost of covering the net $\net$ on day zero is at most $O(\log n)
  \cdot \opt$.
\end{lemma}

The proof of Lemma~\ref{th:dayzero} will occupy the rest of this section; before we do that, note that
Claim~\ref{clm:dayjstar} and  Lemma~\ref{th:dayzero} complete the proof for Theorem~\ref{thm:mult-sc}. Note
that while the definition of the set $\net$ requires us to know $\opt$, we can just run over polynomially many guesses
for $\opt$ and choose the one that minimizes the cost for day zero plus $\tau \cdot k_{j^*}\lambda_{j^*}$
(see~\cite{GNR10-robust} for a rigorous argument).

The proof will show that the fractional cost of covering the elements in the net $N$ is at most $\opt$, and then invoke
the integrality gap for the set covering LP. For the fractional cost, the proof is via a careful backwards induction on
the number of stages, showing that if we mimic the optimal strategy for the first $j-1$ steps, then the fractional cost
of covering the remaining active net elements at stage $j$ is related to a portion of the optimal value as well. This
is easy to prove for the stage $T$, and the claim for stage $0$ exactly bounds the cost of fractionally covering the
net. To write down the precise induction, we next give some notation and formally define what a  strategy
is 
(which will be used in the subsequent sections for the
  other problems as well),
and then proceed with the proof.

\smallskip 
{\bf Formalizing What a Strategy Means} 
For any collection $\G\sse \F$ of sets, let $\cov(\G)\sse U$ denote the elements covered by the sets in $\G$, and let
$c(\G)$ denote the sum of costs of sets in $\G$. At any day $i$, the \emph{state of the system} is given by  the
subsequence $(A_0, A_1, \ldots, A_i)$ seen thus far. Given any scenario sequence  $\A$ and $i\in [T]$, we define $\A_i
= (A_0, A_1, \ldots, A_i)$ to be the partial scenario sequence for days $0$ through $i$.

A solution is a \emph{strategy} $\Phi$, given by a sequence of maps $(\phi_0, \phi_1, \ldots, \phi_T)$, where each one
of these maps $\phi_i$ maps the state $\A_i$ on day $i$ to a collection of sets that are picked on that day. For any
scenario-sequence $\A = (A_1, A_2, \ldots, A_T)$, the strategy $\Phi$ does the following:
\begin{OneLiners}
\item On day $0$, when all elements in $U$ are active,
the sets in $\phi_0$ are chosen, and $\G_1   \leftarrow \phi_0$.
\item At the beginning of day $i\in \{1,\cdots,T\}$, sets in $\G_i$ have
  already been chosen; moreover, the elements in $\cap_{j \leq i} A_j$
  are the active ones. Now, sets in $\phi_i(\A_i)$ are chosen,
  and hence we set $\G_{i+1} \leftarrow \G_i \cup \phi_i(\A_i)$.
\end{OneLiners}

The solution $\Phi = (\phi_i)_i$ is {\em feasible} if for every scenario-sequence $\A = (A_1, A_2, \ldots, A_T)$, the
collection $\G_{T+1}$ of sets chosen at the end of day $T$ covers $\cap_{i \le T} A_i$, i.e.  $\cov(\G_{T+1}) \supseteq
\cap_{i \leq T} A_i$. The cost of this strategy $\Phi$ on a fixed sequence $\A$ is the total effective cost of sets
picked:
\[ \ts C(\Phi\mid\A) = c(\phi_0) + \sum_{i=1}^T \lambda_i \cdot
c\left(\phi_i(\A_i)\right). \] The objective in the robust multistage problem is to minimize $\rcov(\Phi)$, the
effective cost under the {\em worst case} scenario-sequence, namely:
\[ \ts \rcov(\Phi) := \max_{\A} C(\Phi\mid\A)
\]
The goal is to find a strategy with least cost; for the rest of the section, fix $\Phi^* = \{\phi^*_i\}$ to be such a
strategy, and let $\opt = \rcov(\Phi^*)$ denote the optimal objective value.



{\bf Completing Proof of Lemma~\ref{th:dayzero}.}
First, we assume that the inflation factors satisfy $\lambda_{j+1}\ge 12\cdot \lambda_j$ for all $j\ge 0$. If the
instance does not have this property, we can achieve this by merging consecutive days having comparable inflations, and
lose a factor of $12$ in the approximation ratio.
The choice of constant $12$ comes from a lemma from~\cite{GNR10-robust}.
\begin{lemma}[\cite{GNR10-robust}]
  \label{lem:gnr-set-cover}
  Consider any instance of set cover; let $B\in\mathbb{R}_+$ and
  $k\in\mathbb{Z}_+$ be values such that
  \begin{OneLiners}
  \item the set of minimum cost covering any element costs 
    $\geq 36\,\ln m\cdot \frac{B}k$, and
  \item the minimum cost of fractionally covering any $k$-subset of
    elements $\leq B$.
  \end{OneLiners}
  Then the minimum cost of fractionally covering {\bf all} elements is at most
  $r\cdot B$, for a value $r\le 12$.
\end{lemma}

For a partial scenario sequence $\A_i$ on days upto $i$, we use $\phi^*_i(\A_i)$ to denote the sets chosen \emph{on day
$i$} by the optimal strategy, and $\phi^*_{\le i}(\A_i)$ to denote the sets chosen on days $\{0,1,\ldots,i\}$, again by
the optimal strategy.

\begin{definition}\label{def:Vj}
  For any $j\in [T]$ and $\A_j =
  (A_1,\ldots,A_j)$, 
  define
  \[ V_j(\A_j) := \max_{\substack{(A_{j+1}, \cdots, A_T)  \\ |A_{t}| =
      k_{t} \; \forall t }}
  \sum_{i=j}^T r^{i-j} \cdot c\left( \phis_i(\A_i)\right).
  \]
\end{definition}

That is, $V_j(\A_j)$ is the worst-case cost incurred by \Phistar on days $\{j,\ldots,T\}$ conditioned on $\A_j$, under
modified inflation factors $r^{i-j}$ for each day $i\in \{j,\ldots,T\}$.  We use this definition with $r$ being the
constant from Lemma~\ref{lem:gnr-set-cover}. Recall that we assumed that $\lambda_i \geq r^i$.

\begin{fact}
  \label{fct:vzero}
  The function $V_0(\cdot)$ takes the empty sequence as its argument,
  and returns $V_0 = \max_{\A} \sum_{i=0}^T r^{i} \cdot c(\phis_i(\A_i))
  \leq \max_{\A} \sum_{i=0}^T \lambda_i \cdot c(\phis_i(\A_i)) = \opt$,
\end{fact}

For any subset $U'\sse U$ and any collection of sets $\G\sse \f$, define $\lp(U'\mid\G)$ as the minimum cost of
\emph{fractionally covering} $U'$, given all the sets in $\G$ at zero cost. Given any sequence $\A$, it will also be
useful to define $\hatA_j = \cap_{i \leq j} A_j$ as the active elements on day $j$.  Our main technical lemma is the
following:

\begin{lemma}\label{lem:sc-main}
  For any $j\in [T]$ and partial scenario sequence $\A_j$, we have:
  \[ \lp\left( \net\cap \hatA_j \mid \phis_{\le
      j-1}\left(\A_{j-1}\right)\right) \le V_{j}(\A_j).
  \]
  In other words, the fractional cost of covering $\net\cap \hatA_i$
  (the ``net'' still active in stage $j$) given sets
  $\phis_{\le j-1}\left(\A_{j-1}\right)$ for free is at most
  $V_j(\A_j)$.
\end{lemma}

Before we prove this, note that for $j=0$, the lemma implies that $LP(\net) \leq V_0 \leq \opt$. Since the integrality
gap for the set cover LP is at most $H_n$ (as witnessed by the greedy algorithm), this implies that the cost on day~$0$
is at most $O(\log n) \opt$, which proves Lemma~\ref{th:dayzero}.

\begin{proof}
We induct on $j\in\{0,\cdots,T\}$ with $j=T$ as base case. In this
  case, we have a complete scenario-sequence $\A_T = \A$, and the
  feasibility of the optimal strategy implies that $\phis_{\le T}(\A_T)$
  completely covers $\hatA_T$. So,
  \[ \lp\left(\hatA_T \mid \phis_{\le T-1} (\A_{T-1})\right) \le c\left(
    \phis_T(\A_{T}) \right) = V_T(\A_{T}).
  \]
For the induction step, suppose now $j<T$, and assume the lemma for
  $j+1$. Here's the roadmap for the proof: we want to bound the
  fractional cost to cover elements in $\net\cap \hatA_j \setminus
  \cov(\phis_{\le j-1}(\A_{j-1}))$ since the sets $\phis_{\le j-1}(\A_{j-1})$ are free. Some of these elements are covered
  by $\phis_j(\A_j)$, and we want to calculate the cost of the
  others---for these we'll use the inductive hypothesis.
  So given the scenarios $A_1,\ldots,A_{j}$ until day $j$, define
{\small  \begin{equation}\label{eq:theta-defn}
    \theeta_j(\A_j) := \max_{|B_{j+1}|=k_{j+1}}
    V_{j+1}(A_1,\ldots,A_j,B_{j+1})  \implies V_j(\A_{j}) =
    c(\phis_j(\A_{j})) + r\cdot \theeta_j(\A_j).
  \end{equation}}
  Let us now prove two simple subclaims.
  \begin{cl}\label{cl:sc1}
    $\theeta_j(\A_j) \le \opt/\lambda_{j+1}$.
  \end{cl}
  \begin{proofclaim}
Suppose that $\theeta_j(\A_j)$ is defined by the sequence
    $(A_{j+1},\ldots,A_T)$; i.e. $\theeta_j(\A_j) = \sum_{i=j+1}^T r^{i-j-1}
    \cdot c\left(\phis_i(\A_{i})\right)$. Then, considering the
    scenario-sequence $\A = (A_1,\ldots,A_j, A_{j+1},\ldots,A_T)$, we
    have:
{\small    \[ \opt \ge \sum_{i=0}^T \lambda_i\cdot c\left(\phis_i(\A_{i})\right)
    \ge \sum_{i=j+1}^T \lambda_i\cdot c\left(\phis_i(\A_{i})\right) \ge
    \sum_{i=j+1}^T \lambda_{j+1}\, r^{i-j-1} \cdot
    c\left(\phis_i(\A_{i})\right) =\lambda_{j+1}\cdot \theeta_j(\A_j).
    \]}
    The third inequality uses the assumption that $\lambda_{\ell+1}\ge
    r\cdot \lambda_{\ell}$ for all days $\ell$.
  \end{proofclaim}

  \begin{cl}\label{cl:sc2}
    For any $A_{j+1}$ with $|A_{j+1}|=k_{j+1}$, we have
    \[ \lp\left(\net \cap \hatA_{j+1} \mid \phis_{\le j}(\A_{j})\right)
    \le \theeta_j(\A_j). \]
  \end{cl}
  \begin{proofclaim}
    By the induction hypothesis for $j+1$, and $V_{j+1}(\A_{j+1})
    \leq \theeta_j(\A_j)$.
  \end{proofclaim}

Now we are ready to apply Lemma~\ref{lem:gnr-set-cover} to complete the proof of the inductive step.
  \begin{cl}
    \label{cl:sc3}
    Consider the set-system $\g$ with elements $\net' := \net\bigcap \left(
    \hatA_j \setminus \cov(\phis_{\le j}(\A_j) ) \right)$ and the sets
    $\f\setminus \phis_{\le j}(\A_j)$. The fractional cost of
    covering $N'$ is at most $r\cdot \theeta_j(\A_j)$.
  \end{cl}
  \begin{proofclaim}
    In order to use Lemma~\ref{lem:gnr-set-cover} on this set system,
    let us verify the two conditions:
    \begin{enumerate}
    \item Since $N' \sse N$, the cost of the cheapest set covering any $e
      \in N'$ is at least $\tau \ge \beta\cdot
      \frac{\opt}{\lambda_{j+1}\, k_{j+1}}\ge \beta\cdot
      \frac{\theeta_j(\A_j)}{k_{j+1}}$ using the definition of the
      threshold $\tau$, and Claim~\ref{cl:sc1}; recall $\beta = 36 \ln
      m$.
    \item For every $X \sse \net'$ with $|X|\le k_{j+1}$, the minimum
      cost to fractionally cover $X$ in $\g$ is at most
      $\theeta_j(\A_j)$. To see this, augment $X$ arbitrarily to form
      $A_{j+1}$ of size $k_{j+1}$; now Claim~\ref{cl:sc2} applied to
      $A_{j+1}$ implies that the fractional covering cost for $N\cap \hatA_{j+1}=N\cap \left( A_{j+1}\cap \hatA_j\right)$ in $\g$
is at most $\theeta_j(\A_j)$; since $X\sse N'\sse N\cap \hatA_j$ and $X\sse A_{j+1}$ the covering cost for $X$ in $\g$
is also at most $\theeta_j(\A_j)$.
    \end{enumerate}
    We now apply Lemma~\ref{lem:gnr-set-cover} on set-system $\g$ with
    parameters $B:=\theeta_j(\A_j)$ and $k=k_{j+1}$ to infer that the
    minimum cost to fractionally cover $\net'$ using sets from $\f \setminus
    \phis_{\le j}(\A_{j})$ is at most $r\cdot \theeta_j(\A_j)$.
  \end{proofclaim}

  To fractionally cover $N \cap \hatA_j$, we can use the fractional
  solution promised by Claim~\ref{cl:sc3}, and integrally add the sets
  $\phis_{j}(\A_{j})$. This implies that
  \[ \lp\left( \net \cap \hatA_j \mid \phis_{\le j-1}(\A_{j-1}) \right) \le
  c(\phis_j(\A_{j})) + r\cdot \theeta_j(\A_j) = V_{j}(\A_{j}),
  \] where the last equality follows from~(\ref{eq:theta-defn}). This
  completes the induction and proves Lemma~\ref{lem:sc-main}.
\end{proof}

\section{Multistage Robust Minimum Cut}

\def\mc{\ensuremath{{\sf MinCut}}}

We now turn to the multistage robust min-cut problem, and show:
\begin{theorem}\label{thm:mincut}
There is an $O\left(\min\{T,\log n,\log \lambda_{max}\}\right)$-approximation algorithm for $T$-stage $k$-robust minimum cut.
\end{theorem}
In this section we prove an $O(T)$-approximation ratio where $T$ is the number of stages; in \stocoption{ the full version }{ Appendix~\ref{app:scalingT} } we show that simple scaling arguments can be used to ensure $T$ is at most $\min\{\log n,\, \log\lambda_{max}\}$, yielding Theorem~\ref{thm:mincut}. Unlike set cover,
the guarantee here depends on the number of stages. Here is the high-level reason for this additional loss: in each stage of an
optimal strategy for set cover, any element was either completely covered or left completely uncovered---there was no
partial coverage. However in min-cut,
the optimal strategy could keep whittling away at the cut for a node in each stage. The main idea to deal with
this is to use a stage-dependent definition of ``net'' in the inductive proof (see Lemma~\ref{lem:mincut-struct} for
more detail), which in turn results in an $O(T)$ loss.

The input consists of an undirected graph $G=(U,E)$ with edge-costs $c:E\rightarrow \mathbb{R}_+$ and root $\rho$. For
any subset $U'\sse U$ and subgraph $H$ of $G$, we denote by $\mc_H(U')$ the minimum cost of a cut separating $U'$ from
$\rho$ in $H$. If no graph is specified then it is relative to the original graph $G$. Recall that a scenario sequence
$\A = (A_0, A_1, \ldots, A_T)$ where each $A_i \sse U$ and $|A_i| = k_i$, and we denote the partial scenario sequence
$(A_0, A_1, \ldots, A_j)$ by $\A_j$.

We will use notation developed in Section~\ref{sec:set-cover}.  Let the optimal strategy be $\Phi^* =
\{\phis_j\}_{j=0}^T$, where now $\phis_j(\A_j)$ maps to a set of edges in $G$ to be cut in stage~$j$. The feasibility
constraint is that $\phis_{\leq T}(\A_T)$ separates the vertices in $\cap_{i \leq T} A_i$ from the root $\rho$. Let the
cost of the optimal solution be $\opt=\rcov(\Phi^*)$.

Again, the algorithm depends on showing a near-optimal two-stage strategy: define $\tau:= \beta\cdot \max_{j\in [T]}
\frac{\opt}{\lambda_j\,k_j}$, where $\beta = 50$.  Let $j^*=\mbox{argmin}_{j\in [T]} (\lambda_j\,k_j)$. Let the ``net''
$\net:=\{v\in U \mid {\sf MinCut}(v) > 2T\cdot \tau\}$.
The algorithm is:
\begin{quote}
  On day~0, delete $\phi_0 :={\sf MinCut}(\net)$ to separate the ``net''
  $\net$ from $\rho$.

  On day $j^*$, for each vertex $u$ in $A_{j^*}\setminus N$, delete a minimum
  $u$-$\rho$ cut in $G$.


  On all other days, do nothing.
\end{quote}
Again, it is clear that this strategy is feasible: all vertices in $\cap_{i \leq T} A_i$ are either separated from the
root on day $0$, or on day $j^*$. Moreover, the effective cost of the cut on day $j^*$ is at most $\lambda_{j^*}
 \cdot 2T\tau\cdot \abs{A_{j^*}} = 2\beta T\,\opt = O(T)\cdot \opt$. Hence it suffices to show the following:

\begin{lemma}
  \label{thm:mincut-main}
  The min-cut separating $\net$ from the root $\rho$ costs at most
  $O(T) \cdot \opt$.
\end{lemma}

Again, the proof is via a careful induction on the stages.
Loosely speaking, our induction is based on the following: amongst scenario sequences $\A$ containing any fixed ``net''
vertex $v\in N$ (i.e. $v\in \cap_{i \leq T} A_i$) the optimal strategy must reduce the min-cut of $v$ (in an average
sense) by a factor $1/T$ in some stage.

The proof of Lemma~\ref{thm:mincut-main} again depends on a structural lemma proved in \cite{GNR10-robust}:
\begin{lemma}[\cite{GNR10-robust}]
  \label{lem:gnr-min-cut}
  Consider any instance of minimum cut in an undirected graph with root
  $\rho$ and terminals $X$; let $B\in\mathbb{R}_+$ and $k\in\mathbb{Z}_+$
  be values such that
  \begin{OneLiners}
  \item the minimum cost cut separating $\rho$ and $x$ costs 
    $\geq 10\cdot \frac{B}k$, for every $x\in X$.
  \item the minimum cost cut separating $\rho$ and $L$ is $\leq B$, for every
  $L\in {X \choose k}$.
  \end{OneLiners}
  Then the minimum cost cut separating $\rho$ and {\bf all} terminals $X$ is at most $r\cdot
  B$, for a value $r\le 10$.
\end{lemma}
In this section, we assume $\lambda_{j+1}\ge 10\cdot \lambda_j$ for all $j\in [T]$.
Recall the quantity $V_j(\A_j)$ from Definition~\ref{def:Vj}:
\[  V_j(\A_j) := \max_{\substack{(A_{j+1}, \cdots, A_T) \\ |A_{t}| =
      k_{t} \; \forall t }}
  \sum_{i=j}^T r^{i-j} \cdot c\left( \phis_i(\A_i)\right)
\]
where $r:=10$ from Lemma~\ref{lem:gnr-min-cut}. Since $\lambda_i \ge r^i$, it follows that $V_0 \leq \opt$. The next
lemma is now the backwards induction proof that relates the cost of cutting subsets of the net $\net$ to the $V_j$s.
This  finally bounds the cost of separating the entire net $\net$ from $\rho$ in terms of $V_0 \leq \opt$.
Given any $\A_j$, recall that $\hatA_j = \cap_{i \leq j} A_i$.

\begin{lemma}\label{lem:mincut-struct}
  For any $j\in [T]$ and partial scenario sequence $\A_j$,
  \begin{OneLiners}
  \item if $H:=G\setminus \phis_{\le j-1}(\A_{j-1})$ (the residual graph
    in OPT's run at the beginning of stage~$j$), and
  \item $\net_j:= \{v\in \hatA_j \mid\mc_H(v) > (2T-j)\cdot \tau\}$ (the
    ``net'' elements)
  \end{OneLiners}
  then $\mc_H(\net_j)\le 5T\cdot V_{j}(\A_j)$.
\end{lemma}

Before we prove the lemma, note that when we set $j = 0$ the lemma claims that in $G$, the min-cut separating $N_0 = N$
from $\rho$ costs at most $5T\cdot V_0 \leq O(T)\, \opt$, which proves Lemma~\ref{thm:mincut-main}. Hence it
suffices to prove Lemma~\ref{lem:mincut-struct}. Note the difference from the induction used for set-cover: the
thresholds used to define nets is non-uniform over the stages.

\begin{proof}
  We induct on $j\in \{0,\cdots,T\}$. The base case is $j=T$, where we
  have a complete scenario-sequence $\A_T$: by feasibility of the
  optimum, $\phis_{\le T}$ cuts $\hatA_T \supseteq \net_T$ from $r$ in
  $G$. Thus the min-cut for $\net_T$ in $G\setminus \phis_{\le
    T-1}(\A_{T-1})$ costs at most $c(\phis_T(\A_T))=V_T(\A_T)\le 5T\cdot V_T(\A_T)$.

  Now assuming the inductive claim for $j+1\le T$, we prove it for
  $j$. Let $H= G\setminus \phis_{\le j-1}(\A_{j-1})$ be the residual
  graph after $j-1$ stages, and let $\etchj=H\setminus \phis_j(\A_j)$ the
  residual graph after $j$ stages. Let us divide up $\net_j$ into two
  parts, $\net^1_j := \{v\in \net_j \mid \mc_\etchj(v)>(2T-j-1)\cdot \tau\}$
  and $\net^2_j=\net_j\setminus \net^1_j$, and bound the mincut of the
  two parts in $\etchj$ separately.

  \begin{cl} \label{cl:mincut1}$\mc_\etchj(\net^2_j)\le 4T\cdot
    c(\phis_j(\A_j))$.
  \end{cl}
  \begin{proofclaim}
    Note that the set $\net^2_j$ consists of the points that have
    ``high'' mincut in the graph $H$ after $j-1$ stages, but have
    ``low'' mincut in the graph $\etchj$ after $j$ stages. For these we
    use a {\em Gomory-Hu tree}-based argument like that in~\cite{GGR06}.
    Formally, let $t :=(2T-j)\cdot \tau\le 2T\tau$. Hence for every
    $u\in \net^2_j$, we have:
    \begin{equation}
      \ts \mc_H(u) > t \qquad \mbox{and} \qquad \mc_\etchj(u)\le
      \left(1-\frac1{2T}\right) \, t.\label{eq:2}
    \end{equation}
    Consider the Gomory-Hu (cut-equivalent) tree $\tf(\etchj)$ on graph
    $\etchj$, and root it at~$\rho$. For each vertex $u \in \net^2_j$, let $(\cut_u, \overline{\cut}_u)$
    denote the minimum $\rho$-$u$ cut in $\tf(\etchj)$, where $u\in
    \cut_u$ and $\rho\not\in \cut_u$. Pick a subset $\net' \sse \net^2_j$
    such that the union of their respective min-cuts in $\tf(\etchj)$
    separate all of $\net^2_j$ from $\rho$ and their corresponding sets
    $\cut_u$ are disjoint---the set of cuts in tree $\tf(\etchj)$
    closest to the root $\rho$ gives such a collection. 
    Define $F:=\cup_{u\in \net'} \partial_\etchj (\cut_u)$; this is a
    feasible cut in $\etchj$ separating $\net^2_j$ from $\rho$.

    Note that~(\ref{eq:2}) implies that for all $u\in \net^2_j$
    (and hence for all $u \in \net'$), we have
    \begin{OneLiners}
    \item[(i)] $c(\partial_\etchj(\cut_u))\le (1-\frac1{2T})\cdot t$ since $X_u$
      is a minimum $\rho$-$u$ cut in $\etchj$, and
    \item[(ii)] $c(\partial_H(\cut_u))\ge t$ since it is a feasible $\rho$-$u$ cut
      in $H$.
    \end{OneLiners}
    Thus $c(\partial_{H\setminus \etchj}(\cut_u)) = c(\partial_H(\cut_u)) -
    c(\partial_\etchj(\cut_u)) \ge \frac1{2T} \, t \ge \frac1{2T} \cdot
    c(\partial_\etchj(\cut_u))$. So
    \begin{equation}
      c(\partial_\etchj(\cut_u))\le 2T\cdot c(\partial_{H\setminus
        \etchj}(\cut_u)) \quad \mbox{ for all } u\in \net',\label{eq:3}
    \end{equation}
    Consequently,
    \[
    \ts c(F) \le \sum_{u\in \net'} c(\partial_\etchj(\cut_u)) \le 2T
    \cdot \sum_{u\in \net'} c(\partial_{H\setminus \etchj}(\cut_u)) \le
    4T\cdot c(H\setminus \etchj)=4T\cdot c(\phis_j(\A_j)).
    \]
    The first inequality follows from subadditivity,
    the second from~(\ref{eq:3}), the third uses disjointness of
    $\{\cut_u\}_{u\in \net'}$, and the equality follows from $H\setminus
    \etchj = \phis_j(\A_j)$. Thus $\mc_\etchj(\net^2_j)\le 4T\cdot
    c(\phis_j(\A_j))$.
  \end{proofclaim}

Now to bound the cost of separating $\net^1_j$ from $\rho$. Recall the quantity
  $\theeta_j(\A_j)$ from~(\ref{eq:theta-defn}),
  \[
  \theeta_j(\A_j) := \max_{|B_{j+1}|=k_{j+1}}
  V_{j+1}(A_1,\ldots,A_j,B_{j+1}).
  \]
  and that $V_j(\A_j) = \phis_j(\A_j) + r \cdot \theeta_j(\A_j)$.

\begin{cl}
    \label{cl:mincut2}$\mc_\etchj(\net^1_j)\le 5 r\, T \cdot \theeta_j(\A_j)$.
\end{cl}
\begin{proofclaim}
    The definition of $\net^1_j$ implies that for each $u\in \net^1_j$
    we have:
    \begin{equation}\label{eq:mincut-2}
      \mc_\etchj(u) \,\, > \,\, (2T-j-1)\tau \,\, \ge \,\, T\, \tau \,\, \ge \,\, T \cdot \beta\, \frac{\opt}{\lambda_{j+1} \, k_{j+1}} \,\, \ge \,\, \beta T\,
      \frac{\theeta_j(\A_j)}{k_{j+1}},
    \end{equation}
    where the last inequality is by Claim~\ref{cl:sc1}. Furthermore,
    for any $k_{j+1}$-subset $L \sse \net^1_j\sse A_j$ we have:
\begin{equation}\label{eq:mincut-3}
\mc_\etchj(L) \,\, \le \,\, 5T\cdot
      V_{j+1}(A_1,\cdots,A_j,L) \,\, \le \,\, 5T\cdot \theeta_j(\A_j).
\end{equation}
    The first inequality is by applying the induction hypothesis to
    $(A_1,\cdots,A_j,L)$; induction can be applied since $L$ is a ``net'' for this partial scenario sequence (recall $L \sse \net^1_j$ and the
    definition of $\net^1_j$). The second inequality is by
    definition of $\theeta_j(\A_j)$.

    Now we apply Lemma~\ref{lem:gnr-min-cut} on graph $\etchj$ with
    terminals $X=\net^1_j$, bound $B=5T\cdot \theeta_j(\A_j)$, and
    $k=k_{j+1}$. Since $\beta = 50$,
    equations~\eqref{eq:mincut-2}-\eqref{eq:mincut-3} imply that the
    conditions in Lemma~\ref{lem:gnr-min-cut} are satisfied, and we
    get $\mc_\etchj(\net^1_j)\le 5r\, T \cdot \theeta_j(\A_j)$ to prove
    Claim~\ref{cl:mincut2}.
  \end{proofclaim}
  Finally,
  \begin{align*}
    \mc_H(\net_j) & \,\, \le\,\,
    \mc_\etchj(\net_j^1) + \mc_\etchj(\net_j^2)+ c(\phis_j(\A_j)) \\
    & \,\, \le \,\, 5rT \cdot \theeta_j(\A_j) +
    4T\, c(\phis_j(\A_j)) +
    c(\phis_j(\A_j)) \,\, \leq \,\, 5T \cdot V_j(\A_j).
  \end{align*}
  The first inequality uses subadditivity of the cut function, the
  second uses Claims~\ref{cl:mincut1} and~\ref{cl:mincut2}, and the
  third uses $T \geq 1$ and definition of $\theeta_j(\A_j)$.
  This completes the proof of the inductive step, and hence of Lemma~\ref{lem:mincut-struct}.
\end{proof}

\stocoption{}{
\section{Multistage Robust Steiner Tree}

\def\st{{\ensuremath{\sf MinSt}}}
\def\ball{{\ensuremath{\mathsf{B}}}}

We now turn to the multistage robust Steiner tree problem.
\begin{theorem}\label{thm:St-tree}
There is an $O\left(\min\{T,\log n,\log \lambda_{max}\}\right)$-approximation algorithm for $T$-stage $k$-robust Steiner tree.
\end{theorem}
Again, we prove an $O(T)$-approximation ratio where $T$ is the number of stages; Theorem~\ref{thm:St-tree} then follows from the scaling arguments in Appendix~\ref{app:scalingT}. This is the
first of the problems we consider where the net creation is not ``parallel'': whereas in the two previous problems, we
set the net to be all elements that were heavy in some formal sense, here we will pick as our net a set of elements
that are mutually far from each other.

The input consists of an undirected edge-weighted graph $G = (U,E)$.
For any edge-weighted graph $G' = (U', E')$ and $u,v\in U'$, let $d_{G'}(u,v)$ denote the shortest-path distance
between $u$ and $v$ in $G'$; if no graph is specified then the distance is relative to the original $G$. Given a graph
$G' = (U',E')$ and a set of edges $E'' \sse E'$, the graph $G/E''$ is defined by resetting the edges in $E''$ to have
zero length. For any graph $G'$ and subset $X$ of vertices, $\st_{G'}(X)$ is the minimum cost of a Steiner tree on $X$.
Recall the definition of a scenario sequence $\A = (A_0, A_1, \ldots, A_T)$,  partial scenario sequence $\A_j=(A_0,
A_1, \ldots, A_j)$ and  notation from Section~\ref{sec:set-cover}.  The optimal strategy is $\Phi^* =
\{\phis_j\}_{j=0}^T$, where $\phis_j(\A_j)$ maps to a set of edges in $G$ to be chosen in stage~$j$. The feasibility
constraint is that $\phis_{\leq T}(\A_T)$ connect the vertices in $\cap_{i \leq T} A_i$ to each other. The optimal cost
is $\opt=\rcov(\Phi^*)$.

Define $\tau:= \beta\cdot \max_{j\in [T]} \frac{\opt}{\lambda_j\,k_j}$ for $\beta=10$; and $j^*=\mbox{argmin}_{j\in
[T]} (\lambda_j\,k_j)$.  Let $\net$ be a maximal subset of $U$ such that $d(u,v) > 4T\cdot \tau$ for all $u\ne v$,
$u,v\in \net$. The algorithm is:
\begin{quote}
  On day~0, buy edges $\phi_0 :={\sf MST}_G(\net)$ i.e. a 2-approximate Steiner tree.

  On day $j^*$, for each $u \in A_{j^*}$, buy a shortest path from $u$
  to $\net$ in the residual graph $G/\phi_0$.

  On all other days, do nothing.
\end{quote}
Again, the cost incurred on day $j^*$ is not high: by the maximality of $\net$, the distance from any $u \in A_{j^*}$
to $\net$ is at most $4T \tau$, and hence the total effective cost incurred is $4T\tau \cdot |A_{j^*}| \cdot
\lambda_{j^*} = 4T\beta\cdot \opt = O(T)\, \opt$. Thus it is enough to prove the following:

\begin{lemma}
  \label{thm:st-main}
  The MST on the set $\net$ costs at most $O(T) \cdot \opt$.
\end{lemma}

Let us define some useful notation, and review some concepts from previous sections.  For any $v\in U$, distance
$\delta \in \mathbb{R}_+$ and graph $G'$, let the ball $\ball_{G'}(v,\delta)$ denote the set of vertices that are at
most distance $\delta$ from $v$. Given some value $t\ge 0$, a set $N\sse U$ of vertices is called a {\em
$t$-ball-packing} in graph $G'$ if the balls $\left\{\ball_{G'}(u,t)\right\}_{u\in N}$ are disjoint. Finally, recall
$V_j(\A_j)$ as defined in Definition~\ref{def:Vj}; here we set $r=1$. Given $\A$, recall that $\hatA_j = \cap_{i \leq
j} A_i$.

The main structure lemma that will prove Lemma~\ref{thm:st-main} is the following:

\begin{lemma}
  \label{lem:st-struct}
  For any $j\in[T]$ and partial scenario-sequence $\A_j$, if the residual graph
  $H:=G/\phis_{\le j-1}(\A_{j-1})$ and if
  $\net_j\sse \hatA_j$ is {\bf any} $(2T-j)\tau$-ball-packing in $H$ (i.e.,
  the ``net'' elements) then the cost of the minimum Steiner tree
  $\st_H(\net_j)\le 5T\cdot V_j(\A_j)$.
\end{lemma}

Before we prove Lemma~\ref{lem:st-struct}, note that setting $j=0$, the lemma says that the weight of the minimum
Steiner tree in $G$ connecting up elements of $\net$ (which is a $2T\tau$-ball-packing) costs at most $5T\, V_0 \leq
5T\, \opt$, which proves Lemma~\ref{thm:st-main}. One difference from the proof strategy used earlier (set cover
and min-cut) is that we cannot directly rely on structural properties from the two-stage problem~\cite{GNR10-robust},
since this yields only a guarantee exponential in $T$. Instead we give a different self-contained proof of the
inductive step to obtain an $O(T)$ approximation ratio; this also yields an alternate proof of constant approximation
for 2-stage robust Steiner tree.

Now, back to the proof of the lemma.

\begin{proof}
  Let $\tau_i=(2T-i)\tau$ for any $i\in[T]$. In the base case $j = T$,
  we given the complete scenario sequence $\A$. The feasibility of the
  optimal solution implies that the edges in $\phis_{\leq T}(\A)$
  connect up $\hatA_T$, and hence $\net_T$. Consequently, the cost to
  connect up $\net_T$ in $G/\phis_{\leq T-1}(\A_{t-1})$ is at most
  $c(\phis_T(\A_T)) = V_T(\A_T) \leq 5T\, V_T(\A_T)$, since $T \geq 1$.

  For the inductive step, assume that the lemma holds for $j+1$, and let us prove it
  for $j$. Let $H = G/\phis_{\le j-1}(\A_{j-1})$, and $\etchj =
  H/\phis_j(\A_j) = G/\phis_{\le j}(\A_{j})$. Note that by assumption,
  $\net_j$ is a $\tau_j$-ball-packing in graph $H$, so the balls
  $\{\ball_H(u,\tau_j) \mid u\in \net_j\}$ are disjoint.

  Define $\net_j^1:=\{v\in \net_j \mid \ball_\etchj(v,\tau_{j+1}) \sse
  \ball_H(v,\tau_{j}) \}$ be the set of points in $\net_j$ such that
  their $\tau_{j+1}$-balls in $\etchj$ are contained within their
  $\tau_j$-balls in $H$. (Recall that $\tau_{j+1} < \tau_j$, and hence
  this can indeed happen.) Let $\net_j^2=\net_j\setminus \net^1_j$ be
  the remaining points in $\net_j$. Let us prove some useful facts about
  these two sets.

  \begin{cl}\label{cl:st-tree-1}
    $ |\net_j^2|\,\tau_{j+1}\, \le \, 2T\cdot c(\phis_j(\A_j)$.
  \end{cl}
  \begin{proofclaim}
    For any $v\in \net_j^2$, it holds that $\ball_\etchj(v,\tau_{j+1})
    \not \sse \ball_H(v,\tau_{j})$. So there is some vertex $w\in
    \ball_\etchj(v,\tau_{j+1}) \setminus \ball_H(v,\tau_{j})$; this
    means that $d_H(v,w)>\tau_{j}$ but $d_{\etchj}(v,w)\le \tau_{j+1}$;
    the distance to $w$ has shrunk by at least $\tau_j - \tau_{j+1} =
    \tau$. Since $\etchj = H/\phis_j(\A_j)$, this implies that the edges
    bought by $\phis_j(\A_j)$ just within this ball is large. In
    particular, if $E(\ball_H(v,\tau_{j}))$ denotes the edges induced on $\ball_H(v,\tau_{j})$ then:
    \[ c\left(\phis_j(\A_j)\bigcap E(\ball_H(v,\tau_{j}))\right)\ge
    \tau_{j}-\tau_{j+1}=\tau,\quad \mbox{which holds for all }v\in
    \net_j^2.
    \]
    Since the edge-sets $E(\ball_H(v,\tau_{j}))$ for $v\in \net_j^2$ are
    disjoint, we can sum over all vertices in $\net_j^2$ to get
    $c(\phis_j(\A_j))\ge |\net_j^2|\cdot \tau$. Finally, $\tau_{j+1}\le
    2T\cdot \tau$, so $|\net_j^2|\, \tau_{j+1} \le 2T\cdot
    c(\phis_j(\A_j))$ which proves the claim.
  \end{proofclaim}

  \begin{cl}
    \label{clm:st-tree-0}
    The set $\net^1_j$ forms a $\tau_{j+1}$-ball-packing in $\etchj$.
  \end{cl}
  \begin{proofclaim}
    Assume not. Then the $\tau_{j+1}$-balls in $\etchj$ around some two
    points $u,v \in \net^1_j$ must intersect. But these balls are
    contained within the $\tau_j$-balls around them in $H$, by the
    definition of $\net^1_j$, so $\ball_H(u, \tau_j) \cap \ball_H(v,
    \tau_j) \neq \emptyset$. But this contradicts the fact that $u,v \in
    \net_j$, and $\net_j$ was a $\tau_j$-ball-packing in $H$.
  \end{proofclaim}

Recall the definition of $\theeta_j(\A_j)$ as in~\eqref{eq:theta-defn}; note that we set $r = 1$.
  \begin{cl}\label{cl:st-tree-2}
    For any $\tau_{j+1}$-ball-packing $Z \sse\net_j$ in $\etchj$, its
    size $|Z|\le k_{j+1}-1$. Moreover, $\st_\etchj(Z)
    \leq 5T\cdot \theeta_j(\A_j)$.
  \end{cl}
  \begin{proofclaim}
For a contradiction,
    suppose $|Z|\ge k_{j+1}$, and let $A_{j+1}\sse Z$ denote any
    $k_{j+1}$-set.  Observe that $(A_1,\cdots,A_j,A_{j+1})$ is a valid
    partial scenario-sequence.
    Also, $A_{j+1} \sse Z \sse \net_j \sse \hatA_j$, and hence $A_{j+1}
    \sse \cap_{i \leq j+1} A_i=\hatA_{j+1}$. Furthermore,  $A_{j+1}\sse Z$ is a
    $\tau_{j+1}$-ball-packing in $\etchj$.  Thus applying the induction
    hypothesis for $j+1$ on $(A_1,\cdots,A_j,A_{j+1})$, we get that  $\st_\etchj(A_{j+1})\le 5T \cdot V_{j+1}(\A_{j+1})$. Now,
    \[ 5T\cdot V_{j+1}(\A_{j+1}) \,\, \le  \,\, 5T\cdot \theeta_j(\A_{j})  \,\, \le \,\,
    5T\cdot \frac{\opt}{\lambda_{j+1}} \,\, \le  \,\, \frac{5T}\beta\cdot k_{j+1}\,
    \tau  \,\, \le  \,\, \frac{T \, k_{j+1} \,\tau}{2}.
    \]
    The first inequality is by~\eqref{eq:theta-defn} which defines
    $\theeta_j(\A_j)$, the second by Claim~\ref{cl:sc1}, the third by
    the definition of $\tau$, and the last inequality uses $\beta=10$.
    On the other hand, $A_{j+1}$ is a $\tau_{j+1}$-ball-packing in $\etchj$
    and so $\st_\etchj(A_{j+1})> (|A_{j+1}| -1) \tau_{j+1} = (k_{j+1}-1) \cdot
    \tau_{j+1}\ge k_{j+1} T \,\tau/2$, since we may assume $k_{j+1}\ge 2$ (otherwise $k_T\le 1$ and the optimal value is $0$). This contradicts the above
    bound. Thus we must have $|Z|<k_{j+1}$.

    Now augment $Z$ with any $k_{j+1}-|Z|$ elements to obtain $A_{j+1}$,
    and apply the inductive hypothesis again on the scenario-sequence
    $(A_1,\cdots,A_j,A_{j+1})$ and the $\tau_{j+1}$-ball-packing in $H'$, $Z\sse
    \hatA_{j+1}$ to obtain
    \begin{equation}\label{eq:st-5}
      \st_\etchj(Z)  \,\, \le  \,\, 5T\cdot V_{j+1}(\A_{j+1})  \,\, \le  \,\, 5T\cdot \theeta_j(\A_j).
    \end{equation}
    This completes the proof of the claim.
\end{proofclaim}

Construct a maximal $\tau_{j+1}$-ball-packing $Z\sse \net_j$ in the graph $\etchj$ as follows: add all of $\net^1_j$ to
$Z$, and greedily add vertices from $\net_j^2$ to $Z$ until no more can be added without violating the ball-packing
condition. Now, to prove the inductive step, we need to show how to connect up the set $\net_j$ cheaply in the graph
$H$. We first bound the cost of this Steiner tree in $\etchj$. Claim~\ref{cl:st-tree-2} says we can connect up $Z$ in
the graph $\etchj$ at cost $5T\, \theeta_j(\A_j)$. Since $Z$ is a maximal $\tau_{j+1}$-ball-packing inside $\net_j$, we
know that each element in $\net_j \setminus Z \sse \net_j^2$ is at distance at most $2\tau_{j+1}$ from some element of
$Z$. There are only $2T c(\phis_j(\A_j))/\tau_{j+1}$ elements in $\net_j^2$ by Claim~\ref{cl:st-tree-1}, so the total
cost to connect these points to $Z$ is at most their product $4T c(\phis_j(\A_j))$, giving us that
\[ \st_\etchj(\net_j) \leq 5T\cdot
\theeta_j(\A_j) + 4T\cdot c( \phis_j(\A_j))
\]
Finally, since the length of the Steiner-tree$(N_j)$ in $H$ can only be greater by $\phis_j(\A_j)$, we get that
\[ \st_H(\net_j) \leq 5T\cdot
\theeta_j(\A_j) + 5T\cdot c( \phis_j(\A_j)) = 5T\cdot V_j(\A_j),
\]
which completes the proof of the lemma.
\end{proof}

\section{Multistage Robust Steiner Forest}

\def\stf{{\ensuremath{\sf MinSF}}}
\def\sfnet{\textsf{SFnet}\xspace}
\def\AlgE{E_{\textsf{alg}}}

Here we consider the multistage robust version of the Steiner Forest problem:
\begin{theorem}\label{thm:St-forest}
There is an $O\left(\min\{T,\log n,\log \lambda_{max}\}\right)$-approximation algorithm for $T$-stage $k$-robust Steiner forest.
\end{theorem}
Again, we prove an $O(T)$-approximation ratio where $T$ is the number of stages; Theorem~\ref{thm:St-forest} then follows from the scaling arguments in Appendix~\ref{app:scalingT}. Here, both the algorithm and the proof will be slightly more involved than the previous problems. Recall that
the input for Steiner Forest consists of an undirected edge-weighted graph $G$, and pairs $\{s_i,t_i\}_{i\in P}$. For
any graph $H$ and subset $P'\sse P$ of pairs, we let $\stf_H(P')$ denote the minimum cost of a Steiner forest
connecting pairs in $P'$; again if the graph is not specified it is relative to the input graph $G$.

For $\Delta \geq 0$, we define a subset $\net\sse P$ of pairs to be a {\em
  $\Delta$-\sfnet} in graph $G'$ if
\begin{OneLiners}
\item $d_{G'}(s_i,t_i) > \Delta$ for all $i\in \net$, and
\item there exist $z_i\in \{s_i,t_i\}$ for all $i\in \net$ such that
  the balls $\left\{\ball_{G'}(z_i,\Delta)\right\}_{i\in \net}$ are disjoint.
\end{OneLiners}

The following simple property immediately follows from the definition:
\begin{lemma}\label{lem:st-forest-net}
  If $\net\sse P$ is a $\Delta$-\sfnet in graph $G'$ then
$\stf_{G'}(\net)\ge |\net|\cdot \Delta$.
\end{lemma}
\begin{proof}
  Consider the dual of the Steiner forest LP relaxation, which is a
  packing problem. Each ball of radius at most $\Delta$ around any vertex in
  $\{z_i\}_{i\in \net}$ is a feasible variable in the dual, since
  $d_{G'}(s_i,t_i) > \Delta$ for all $i\in \net$. Now since
  $\left\{\ball_{G'}(z_i,\Delta)\right\}_{i\in \net}$ are disjoint, there is a 
  feasible dual solution of value $|\net|\cdot \Delta$. Hence the optimal Steiner forest costs at least as much.
\end{proof}

However, it seems quite non-trivial to even compute a maximal $\Delta$-\sfnet. This is unlike the previous algorithms
(set cover, min-cut, Steiner tree) where computing this net was straightforward. Instead, we will run a specialized
procedure to compute a near-maximal net which suffices to give an algorithm for multistage Steiner forest.

Now to describe the algorithm. Define $\beta := 10$, $\tau:= \beta\cdot \max_{j\in [T]} \frac{\opt}{\lambda_j\,k_j}$ and
$j^*=\mbox{argmin}_{j\in [T]} (\lambda_j\,k_j)$. We now run Algorithm~\ref{alg:sfnet} below to
find a $\gamma=2T\tau$-\sfnet $\net \sse P$, as well as a set of edges $\AlgE \sse E$. Then,
\begin{quote}
  On day~0, buy the edges in $\phi_0 := \AlgE$.

  On day $j^*$, for each $(s,t) \in A_{j^*}$, buy a shortest
  $(s,t)$-path in the residual graph $G/\phi_0$.

  On all other days, do nothing.
\end{quote}

Set $\gamma := 2T\tau$. In Algorithm~\ref{alg:sfnet},  $G/(S_r\cup S_f)$ denotes the graph obtained from $G$ by
identifying all pairs in $S_r$ and $S_f$. We note that this algorithm is essentially same as the one for 2-stage robust
Steiner forest in~\cite{GNR10-robust}; however we need to rephrase slightly in order to fit it into our context.
\begin{algorithm}
  \caption{Algorithm for near-maximal $\gamma$-\sfnet for Steiner forest}
  \begin{algorithmic}[1]
  \label{alg:sfnet}
    \STATE \textbf{let} $S_r, S_g, S_o, S_b, S_f, W \gets \emptyset$.

    \WHILE{there exists a pair $i \in P$ with $d_{G / (S_r \cup S_f)}(s_i,t_i) > 4\, \gamma$}

    \STATE \textbf{let} $S_r \gets S_r \cup \{i\}$

    \STATE \label{step:stf-net-1} \textbf{if} $d_G(s_i, w) < 2\,\gamma$ for
    some $w \in W$ \textbf{then} $S_f \gets S_f \cup \{(s_i,w)\}$
    \textbf{else} $W \gets W \cup \{s_i\}$ \label{sf-algo1}

    \STATE \label{step:stf-net-2} \textbf{if} $d_G(t_i, w') < 2\,\gamma$ for
    some $w' \in W$ \textbf{then} $S_f \gets S_f \cup \{(t_i,w')\}$
    \textbf{else} $W \gets W \cup \{t_i\}$ \label{sf-algo2}

    \STATE \textbf{let} $\delta_i\in\{0,1,2\}$ be the increase in $|W|$ due to steps~\eqref{step:stf-net-1}
    and~\eqref{step:stf-net-2}.
    \STATE \textbf{if} $(\delta_i=0)$ \textbf{then} $S_b\gets S_b\cup \{i\}$; \textbf{if} $(\delta_i=1)$ \textbf{then} $S_o\gets S_o\cup
    \{i\}$; \textbf{if} $(\delta_i=2)$ \textbf{then} $S_g\gets S_g\cup \{i\}$.
    \ENDWHILE
   \STATE \textbf{let} $\AlgE:=$  2-approximate Steiner forest on pairs $S_r$,  along with shortest-paths connecting every pair in $S_f$.
   \STATE \textbf{output} the $\gamma$-\sfnet $\net:=S_g\bigcup S_o$,
   and edge set $\AlgE$.
  \end{algorithmic}
\end{algorithm}

To bound the cost incurred on day $j^*$, we note that the algorithm adds pairs in $P$ whose distance in the residual
graph $G/\AlgE$ is more than $4\gamma = 8T\tau$. Hence the cost of connecting any pair in $P \setminus \net$ is at most
$8T\tau$, and consequently the total effective cost incurred on day $j^*$ is $8T\tau \cdot |A_{j^*}| \cdot
\lambda_{j^*} = 8T\beta\cdot \opt = O(T)\, \opt$. Thus it is enough to prove the following lemma, whose proof we
present in the next section:

\begin{lemma}
  \label{thm:stf-main}
  The cost of the edges $\AlgE$ is at most $144T\, \opt$.
\end{lemma}

The proof bounding the cost of edges in $\AlgE$ is also more complicated that in previous sections; here's the roadmap.
First, 
we use the (by now familiar) reverse inductive proof to bound the cost of edges
in an optimal Steiner Forest connecting the pairs in the net $\net$. However, our set of edges $\AlgE$ is not just an
approximate Steiner Forest for the pairs in $\net$, but also connects other nodes, and we have to also bound the cost of the
extra edges our algorithm buys.

For the rest of the section, we use $\gamma := 2T\tau$. Recall that $\hatA_j = \cap_{i \leq j} A_j$. Also recall
$V_j(\A_j)$ from Definition~\ref{def:Vj}; here we use $r=2$. Recall that a set $S \sse U$ in a graph $G' = (U,E')$ is
called a \emph{$t$-ball-packing} if the balls $\{\ball_{G'}(x, t)\}_{t \in S}$ are disjoint. Hence, the second
condition in the definition of a set $\net$ being a $\Delta$-\sfnet is that for each $\{s_i, t_i\} \in \net$, there
exists $z_i \in \{s_i, t_i\}$ such that the set of these $z_i$'s is a $\Delta$-ball-packing. If $\net$ is a \sfnet, let
$\zeta(\net)$ denote this ball-packing that witnesses it.


\begin{lemma}\label{lem:stf-struct}
  For any $j\in[T]$ and partial scenario-sequence $\A_j$, let the residual graph in stage $j$ of the optimal algorithm be 
  $H:=G/\phis_{\le j-1}(\A_{j-1})$, and let $\net_j\sse \hatA_j$ be {\bf any}
  $(2T-j)\tau$-\sfnet in $H$ (i.e., the ``net'' elements). Then the
  cost of the optimal Steiner forest on $\net_j$ is $\stf_H(\net_j)\le 9T\cdot
  V_j(\A_j)$.
\end{lemma}

\begin{proof}
Let $\tau_i=(2T-i)\tau$ for any $i\in[T]$. In the base case $j = T$,
  we are given the complete scenario sequence $\A$. The feasibility of the
  optimal solution implies that the edges in $\phis_{\leq T}(\A)$
  connect the pairs $\hatA_T$, and hence $\net_T$. Consequently, the cost to
  connect up $\net_T$ in $G/\phis_{\leq T-1}(\A_{t-1})$ is at most
  $c(\phis_T(\A_T)) = V_T(\A_T) \leq 9T\cdot V_T(\A_T)$.

Let $\tf_j:=\{s_i, t_i  \mid i\in \net_j\}$ be all the terminals in $\net_j$. Since $\net_j$
  is a $\tau_j$-\sfnet in graph $H$, $\zeta(\net_j)$ is a
  $\tau_j$-ball-packing in $H$. Let graph $\etchj:=H/\phis_j(\A_j)$.
  Similar to the proof for Steiner tree, define $\net^1_j := \{ i \in \net_j \mid \ball_\etchj(z_i,\tau_{j+1}) \sse
  \ball_H(z_i,\tau_{j}) \}$, and let $\net^2_j=\net_j\setminus \net^1_j$
  be the rest of the pairs. The fact that $\net^1_j$ is a
  $\tau_{j+1}$-\sfnet in $\etchj$ follows from the definition of
  $\net^1_j$ and Claim~\ref{clm:st-tree-0}.
  We bound the costs $\stf_\etchj(\net^1_j)$ and $\stf_\etchj(\net^2_j)$
  separately by invoking the inductive hypothesis twice. Let
  $\theeta_j(\A_j)$ be defined as in~\eqref{eq:theta-defn}; note that
  $V_j(\A_j) = c(\phis_j(\A_j) + 2\, \theeta_j(\A_j)$, since $r = 2$ in
  this case.

  The next two claims have proofs almost identical to the Steiner tree
  Claim~\ref{cl:st-tree-2}, and Claim~\ref{cl:st-tree-1} respectively.

  \begin{cl}\label{cl:st-forest-1}
    The cardinality of $\net^1_j$ is $|\net^1_j|< k_{j+1}$, and
    $\stf_\etchj(\net^1_j)\le 9T\cdot \theeta_j(\A_j)$.
  \end{cl}

  \begin{cl}\label{cl:st-forest-2}
    We have $|\net^2_j|\,\tau_{j+1}\le 2T\cdot c(\phis_j(\A_j))$.
  \end{cl}
 
Recall we need to show that the cost to connect up pairs in $\net$ is
  small; we now have that connecting up $\net^1_j$ is not expensive, and have
  a bound on the cardinality of $\net^2_j$---it remains to show that
  this can be used to bound its cost. This we do in the following
  discussion culminating in Claim~\ref{cl:st-forest-5}.

  Let $\tf_j^2 := \{s_i,t_i \mid i\in \net^2_j \}$ be the terminals
  corresponding to pairs in $\net_j^2$.  Consider an auxiliary graph
  $\g_1$ on these terminals, where edges connect nodes at distance at
  most $2\tau_{j+1}$---i.e., $V(\g_1) = \tf^2_j$, and $E(\g_1) = \{(a,b)
  \mid \, a,b\in \tf^2_j,\, d_\etchj(a,b)\le 2\tau_{j+1} \}$. Let
  $C_1,\ldots,C_p$ denote the connected components in this graph $\g_1$,
  so $\tf^2_j=\cup_{\ell=1}^p \, C_\ell$. See also Figure~\ref{fig:StF}. Let $\m$ denote the minimum
  length forest in the graph $\etchj$ having the same connected
  components as $\g_1$.
  \begin{cl}\label{cl:st-forest-3}
    The cost of the forest $\m$ (in graph $\etchj$) is at most $8T\cdot
    c(\phis_j(\A_j))$.
  \end{cl}
  \begin{proofclaim}
    Two nodes in $\g_1$ are connected only if their distance in $\etchj$
    is at most $2\tau_{j+1}$; moreover, the number of vertices in $\g_1$
    is $|\tf^2_j|= 2|\net^2_j|$. Hence, the forest $\m$ costs at
    most $4|\net^2_j| \tau_{j+1}$, which is at most $8T\;
    c(\phis_j(\A_j))$ using Claim~\ref{cl:st-forest-2}.
  \end{proofclaim}

  Define $\alpha:\tf^2_j\rightarrow \{1, \ldots, p\}$ mapping each
  terminal in $\tf^2_j$ to the index for its component in $\g_1$: i.e.,
  $\alpha(u)=\ell$ if $u\in C_\ell \sse \g_1$. Define another auxiliary
  multigraph $\g_2$ on vertices $\{1,2, \ldots, p\}$ and edges
  $E(\g_2):= \{(\alpha(s_i),\alpha(t_i))\mid i\in \net^2_j \}$. (See Figure~\ref{fig:StF}.) Observe
  that there is a one-to-one correspondence between edges in $\g_2$ and
  pairs in $\net^2_j$. Let $\f$ denote the edges of any maximal forest
  in $\g_2$ ($\f$ cannot contain any self-loops); we also use $\f$
  to denote the corresponding pairs from $\net^2_j$.

\begin{figure}
\begin{center}
\includegraphics[scale=0.72]{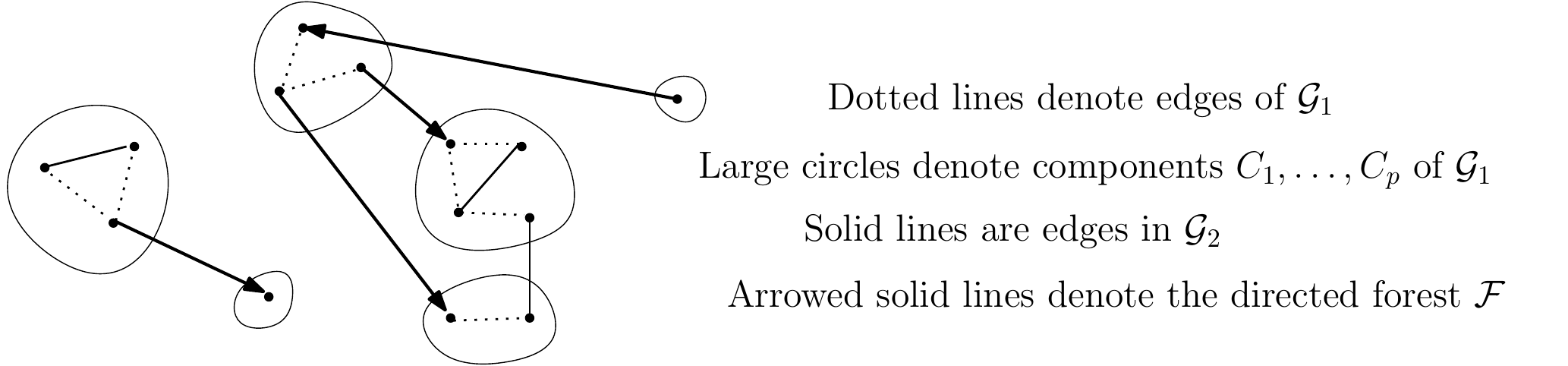}
\caption{The auxiliary graphs $\g_1$ and $\g_2$. \label{fig:StF}}
\end{center}
\end{figure}

  \begin{cl} \label{cl:st-forest-4}
    $\stf_\etchj(\f)\le 9T\cdot \theeta_j(\A_j)$.
  \end{cl}
  \begin{proofclaim}
Orient the edges in the forest $\f$ so that each vertex $\ell\in
    V(\g_2)$ has at most one incoming edge; call this map $\pi: \f \to
    V(\g_2)$, and note that for each $\ell \in \f$, $\pi^{-1}(\ell) \leq
    1$.  For each $i \in \f$, set $z'_i = s_i$ if the edge is
    oriented towards $s_i$, and $t_i$ otherwise. Clearly, the vertices
    $\{z'_i\}_{i\in\f}$ lie in distinct components of $\g_1$. Hence
    $\{z'_i\}_{i\in\f}$ is an independent set in $\g_1$, and hence is a
    $\tau_{j+1}$-ball-packing in $\etchj$. Moreover, since $\{s_i,
    t_i\}$ lie in distinct components of $\g_1$, $d_\etchj(s_i,t_i) >
    2\tau_{j+1}$; hence $\f$ is also a $\tau_{j+1}$-\sfnet in graph
    $\etchj$.  Now, as in the proof of Claim~\ref{cl:st-tree-2}, it
    follows by induction on $\f$ that $\stf_\etchj(\f)\le 9T\cdot
    \theeta_j(\A_j)$.
  \end{proofclaim}

  \begin{cl}
    \label{cl:st-forest-5} $\stf_\etchj(\net^2_j)\le \stf_\etchj(\f) +
    c(\m) \le 9T\cdot \theeta_j(\A_j) + 8T\cdot c(\phis_j(\A_j))$.
  \end{cl}
  \begin{proofclaim}
    We will show that for every $i\in \net^2_j$, the subgraph
    $\m\bigcup\stf_\etchj(\f)$ connects $s_i$ to $t_i$ in graph
    $\etchj$, after which the claim follows directly from
    Claims~\ref{cl:st-forest-3} and~\ref{cl:st-forest-4}.

    Recall that $\m$ has the same connected components as graph
    $\g_1$. For any $i\in N^2_j$, let $a=\alpha(s_i)$ and $b=\alpha(t_i)$ be the
    indices of the components containing $s_i$ and $t_i$
    respectively. If $a=b$, then $\m$ connects $\{s_i,t_i\}$, so assume
    $a \neq b$. Since $\f$ is a maximal forest in $\g_2$, it contains
    some path $e_1,\ldots,e_q$ from $\alpha(a)$ to $\alpha(b)$. But now
    $\stf_\etchj(\f) \cup \m$ connects $s_i$ to $t_i$.
  \end{proofclaim}

  Finally using sub-additivity of Steiner Forest, and
  Claims~\ref{cl:st-forest-1}
  and~\ref{cl:st-forest-5},
  \[
  \stf_\etchj(\net_j)\le
  \stf_\etchj(\net^1_j) + \stf_\etchj(\net^2_j) \le 2\cdot 9T\cdot
  \theeta_j(\A_j) + 8T\cdot c(\phis_j(\A_j)).
  \]
  Thus $\stf_H(\net_j)\le c(\phis_j(\A_j)) + \stf_\etchj(\net_j)\le
  9T\cdot [ c(\phis_j(\A_j)) + 2\cdot \theeta_j(\A_j) ]=9T\cdot
  V_j(\A_j)$.
\end{proof}


As a consequence of the case $j=0$ of Lemma~\ref{lem:stf-struct} from the previous section, we know that the cost of
the optimal Steiner forest on any $\tau_0 = 2T\tau$-\sfnet $\net_0$ in the original graph $G$ is at most $9T\, V_0\leq
9T\, \opt$. And indeed, the set $N = S_g\cup S_o$ is a  $2T\tau$-\sfnet in $G$, since $\{\ball_G(w,2T\tau) \mid w\in
W\}$ are disjoint and $W\cap \{s_i,t_i\}\ne \emptyset$ for all $i\in S_g\cup S_o$. However, the set $\AlgE$ is not just
the Steiner forest on $S_g \cup S_o$, it is actually a Steiner forest on $S_r = S_g \cup S_o \cup S_b$, along with
shortest paths between every ``fake'' pair in $S_f$. This is what we bound in the proof below.

\begin{proofof}{Lemma~\ref{thm:stf-main}}
  Observe that $c(\AlgE)\le 2\cdot \stf(S_r) + 2\gamma\cdot |S_f|$,
  since the distance between each pair in $S_f$ is at most
  $2\gamma$. (Remember, $\gamma = 2T\tau$.) Recall that $\net =S_g\cup
  S_o$ and $S_r=\net \cup S_b$; so $\stf(S_r)\le
  \stf(\net)+\stf(S_b)$. Finally, the argument in the previous paragraph
  shows that $\stf_G(N)\le 9T\, V_0\le 9T\, \opt$. Hence
  \begin{equation}\label{eq:stf-algo1}
    c(\AlgE)\le 2\cdot9T\,\opt + 2\cdot \stf(S_b) + 2\gamma\cdot |S_f|.
  \end{equation}
  Moreover, $S_b$ might be very far from a $\gamma$-\sfnet, so we cannot
  just apply the same techniques to it.

  \paragraph{Bounding $\stf(S_b)$.} As in the proof of
  Lemma~\ref{lem:stf-struct}, define an auxiliary graph $\g_1$ with
  vertices $V(\g_1) = \{s_i,t_i\mid i\in S_b\}$, and edges $E(\g_1) =
  \{(a,b) \mid \, d_G(a,b)\le 2\gamma\}$ between any two terminals in
  $V(\g_1)$ that are at most $2\gamma$ apart. Let $C_1,\ldots,C_p$ be the
  connected components in this graph $\g_1$, and let $\m$ be the minimum
  cost spanning forest in $G$ having the same connected components as
  $\g_1$. Since there are $2|S_b|$ vertices in $\g_1$ and edges correspond
  to pairs at most $2\gamma$ from each other, the cost $c(\m)\le
  4|S_b|\, \gamma$.

  Again, define map $\alpha:V(\g_1)\rightarrow \{1, \ldots, p\}$ where
  $\alpha(u)=\ell$ if $u\in C_\ell$. Define another auxiliary graph
  $\g_2$ on vertices $\{1, 2, \ldots, p\}$ with edges $E(\g_2):=
  \{(\alpha(s_i),\alpha(t_i))\mid i\in S_b\}$. Let $\f$ denote the edges
  of any maximal forest in $\g_2$, and also to denote the corresponding
  pairs from $S_b$. By orienting $\f$ so that each vertex has indegree
  at most one, we obtain $z_i'\in\{s_i,t_i\}$ for all $i\in \f$
  satisfying $|C\cap \{z_i'\}_{i\in \f}|\le 1$ for each component $C$ of
  $\g_1$ (see Claim~\ref{cl:st-forest-4} for details). Thus
  $\{z_i\}_{i\in \f}$ is an independent set in $\g_1$, and so
  $\{\ball_G(z_i,r) \mid i\in\f\}$ are disjoint.  This implies that $\f$
  is a  $\gamma$-\sfnet in $G$. Applying Lemma~\ref{lem:stf-struct} on
  $\f$ (with $j=0$) gives $\stf_G(\f)\le 9T\cdot\opt$.

  Finally, as in Claim~\ref{cl:st-forest-5} we obtain $\stf(S_b)\le
  \stf(\f) + c(\m)\le 9T\cdot\opt+4|S_b|\, \gamma$.  Combining this
  with~\eqref{eq:stf-algo1}  we have:
  \begin{equation}\label{eq:stf-algo3}
    c(\AlgE) \le 36T\cdot\opt+\,
    8|S_b|\,\gamma+\, 2|S_f|\,\gamma.
  \end{equation}

  \paragraph{Bounding $|S_b|$ and $|S_f|$.}
  We use the following property of Algorithm~1.
  \begin{cl}[\cite{GNR10-robust}] In any execution of Algorithm~1, $|S_b|\le
    |\net|$ and $|S_f|\le 2|\net|$.
  \end{cl}
  \begin{proofclaim}
    This follows directly from the analysis
    in~\cite{GNR10-robust}. Lemma~5.3 in that paper yields $|S_f|\le
    |S_r|$. By the definition of various sets $S_g,S_b,S_o,S_f$ and
    $\net = S_g \cup S_o$, we have that $|S_f|\ge 2|S_b|$ and
    $|S_r|=|\net|+|S_b|$. Thus we have $2|S_b|\le |S_f|\le |S_r|\le
    |\net|+|S_b|$, implying $|S_b|\le |\net|$. Finally, $|S_f|\le
    |S_r|=|\net|+|S_b|\le 2|\net|$.
  \end{proofclaim}
Combining this with~\eqref{eq:stf-algo3}, it just remains to bound
  $|\net|\,\gamma$. Recall that $\net$ is a  $\gamma$-\sfnet in $G$; so
  $\gamma\,|\net| \leq \stf_G(\net)$. On the other hand, we already
  argued that $\stf_G(\net)\le 9T\,\opt$.  Putting this together
  with~(\ref{eq:stf-algo3}), we get
  \[ c(\AlgE) \le 36T\cdot\opt+\, 8|S_b|\, \gamma+\, 2|S_f|\, \gamma \le
  36T\cdot\opt+\, 12|\net|\, \gamma\le 144T\cdot\opt.
  \]
  This completes the proof of Lemma~\ref{thm:stf-main}.
\end{proofof}
}

\paragraph{Acknowledgments.} We thank F.~B.~Shepherd, A.~Vetta, and
M.~Singh for their generous hospitality during the initial stages of
this work.

\stocoption{
\bibliographystyle{plain}
\bibliography{robust,../abbrev,../my-papers,../embedding}
}{
\bibliographystyle{plain}
{\small \bibliography{robust,../abbrev,../my-papers,../embedding}}
}


\stocoption{}{
\appendix

\section{Preprocessing to Bound Number of Stages $T$} 
\label{app:scalingT}

We show using a scaling argument how to ensure $T\le O(\min\{\log \lambda_{max},\, \log n\})$ for the
min-cut, Steiner tree and Steiner forest problems; here $n$ is the number of vertices in the input graph). Recall that
the approximation for multistage robust set cover does not depend on $T$, hence such a preprocessing step is not
required for set cover.

It is easy to ensure that $\lambda_{max}=\lambda_{T}\ge 2^T$ losing a constant factor in the objective: we simply find
a maximal subsequence of stages $i$ where the inflation $\lambda_i$ in each stage in the subsequence increases by at
least a factor of two from the previous stage in the subsequence, and perform actions only on these stages (or in the
very final stage $T$). This means $T\le \log_2 \lambda_{max}$.

We now ensure that $T\le O(\log n)$. Guess the edge $f$ of maximum cost (unscaled) that is ever bought by the optimal strategy on any scenario-sequence;
there are at most $n^2$ choices for $f$ and the algorithm will enumerate over these. It follows that $\opt\ge c_f$
since $f$ must be bought on some scenario-sequence, and all inflation factors are $\ge 1$. Let $E_{high}:=\{e\in E:
c_e>c_f\}$; by the choice of $f$, no edge in $E_{high}$ is ever used by \opt for any scenario-sequence. So the optimal
value does not change even if we remove edges $E_{high}$ from the instance---in Steiner tree/forest this means deleting
edges $E_{high}$ from the graph, and in min-cut this means contracting edges $E_{high}$. This yields an instance with
maximum edge-cost $c_{max}\le c_f$.

Let $E_{low}:=\{e\in E: c_e< c_f/n^2\}$; the total cost of edges in $E_{low}$ is at most $|E_{low}| \cdot
\frac{c_f}{n^2}\le c_f\le \opt$ by the choice of $f$. So we can assume that all edges in $E_{low}$ are always bought in
stage $0$, at the loss of an additive \opt term in the objective. This ensures that the minimum edge-cost $c_{min}\ge
c_f/n^2$. Combined with the above step, we have an instance with $\frac{c_{max}}{c_{min}}\le n^2$.

Observe that any stage $i$ with $\lambda_i > n^2\cdot \frac{c_{max}}{c_{min}}$ is completely inactive in any optimal
strategy: if not then the objective of the resulting strategy is greater than $n^2\cdot c_{max}$, whereas the trivial
strategy of buying all elements in stage $0$ costs at most $n^2\cdot c_{max}$. Hence, without loss of generality, we have $\lambda_{max}\le
n^2\cdot \frac{c_{max}}{c_{min}}$, which combined with the above gives $\lambda_{max}\le n^4$. Thus $T\le \log_2
\lambda_{max}\le O(\log n)$.


\section{Some Useful Examples}
\subsection{Non-Optimality of Two-Stage Strategies for Set Cover} \label{sec:lowerbound}

We give an instance of multistage set-cover where any optimal solution has to buy sets on all days. This shows that we
really need to consider near-optimal strategies to prove our structure result about thrifty strategies.

Let $T$ denote the time horizon. The scenario bounds are $k_i=T+1-i$ for each $0\le i\le T$. The inflation factors are
$\lambda_i=(1+\epsilon)^i$ for each $i\in [T]$, where $\epsilon>0$ is chosen such that $\lambda_T<2$. The elements in
the set-system are $U:=\{1,\cdots,T,T+1\}$. The sets and their costs are as follows:
\begin{OneLiners}
\item For each $i\in \{1,\cdots,T-1\}$, define $\s_i$ as the collection consisting of all $k_i=T+1-i$ subsets of $\{i,\cdots,T+1\}$ {\em except}
$\{i+1,\cdots,T+1\}$; each of the sets in $\s_i$ has cost $\lambda_T/\lambda_i$. Note that each set in $\s_i$
contains $i$.
\item There are two singleton sets $\{T\}$ and $\{T+1\}$ of cost one each; let $\s_T:=\{ \{T\}, \{T+1\} \}$.
\end{OneLiners}
The number of sets is at most $T^2$, and each costs at least one.

Consider the strategy $\sigma$ that does nothing on day $0$ and for each day $i\in\{1,2,\ldots,T\}$ does:
\begin{OneLiners}
\item If the scenario $A_i\in {U\choose k_i}$ in day $i$ equals one of the sets in $\s_i$ then buy set $A_i$ on day $i$. Note that this set $A_i$ in the
set-system has cost $\lambda_T/\lambda_i$.
\item For any other scenario $A_i\in {U\choose k_i}$, buy nothing on day $i$.
\end{OneLiners}
It can be checked directly that this is indeed a feasible solution. Moreover, the effective cost under every
scenario-sequence is exactly $\lambda_T<2$. Note that for every scenario-sequence, $\sigma$ buys sets on exactly one
day; however the days corresponding to different scenarios are different. In fact, for each $i\in [T]$ there is a
scenario-sequence (namely  $A_j=\{j+1,\cdots,T+1\}$ for $j<i$ and $A_i\in \s_i$) where $\sigma$ buys sets on day $i$.
Thus strategy $\sigma$ buys sets on all days.

We now claim that  $\sigma$ is the unique optimal solution to this instance. For any other feasible strategy $\sigma'$,
consider a scenario-sequence $(A_1,\cdots,A_T)$ where  $\sigma'$ behaves differently from $\sigma$. Let
$i\in\{0,\cdots,T\}$ denote the earliest day when  $\sigma'$ differs from  $\sigma$ under this scenario-sequence. There
are two possibilities:
\begin{OneLiners}
\item $A_j=\{j+1,\cdots,T+1\}$ for each $j<i$ and $A_i\in \s_i$, i.e. $\sigma$ buys nothing before day $i$ and buys set $A_i$ on day
$i$. $\sigma'$ buys nothing before day $i$ and covers $A'_i\subsetneq A_i$ on day $i$. 
From the construction of sets, it follows that $i\not\in A'_i$ (otherwise $\sigma'$ buys some set of cost $\ge
\lambda_T/\lambda_i$ and so we may assume $A'_i=A_i$). Consider any scenario-sequence that equals $A_1,\cdots,A_i$ until day $i$,
and has $\{i\}$ as the scenario on day $T$: the effective cost under $\sigma'$ is then at least $\lambda_{i+1} \cdot
\lambda_T/\lambda_i > \lambda_T$ (since the min-cost set covering $i$ costs $\lambda_T/\lambda_i$ and faces inflation at
least $\lambda_{i+1}$) .

\item $A_j=\{j+1,\cdots,T+1\}$ for each $j\le i$, i.e. $\sigma$ buys nothing until day $i$ (inclusive). $\sigma'$ buys nothing before day $i$ and
covers $A'_i\ne\emptyset$ on day $i$. We have the two cases:
\begin{enumerate}
 \item If $A'_i\subsetneq A_i$ then let $e\in A_i\setminus A'_i$. Consider any scenario-sequence that equals $A_1,\cdots,A_i$
until day $i$, and has $\{e\}$ as the scenario on day $T$: the effective cost under $\sigma'$ is then at least $2\cdot
\lambda_i > \lambda_T$ (since $\sigma'$ buys at least two sets on days $i$ or later).
 \item If $A'_i=A_i=\{i+1,\ldots,T+1\}$ then the cost of $\sigma'$ under any scenario-sequence with $A_1,\cdots,A_i$
until day $i$ is at least $2\cdot \lambda_i > \lambda_T$ (since the min set cover of $\{i+1,\ldots,T+1\}$ costs at least $2$).
\end{enumerate}
\end{OneLiners}
In all cases above $\sigma'$ has objective value strictly larger than $\lambda_T$. Thus $\sigma$ is the unique optimal
solution.

\medskip
{\bf Remark:} It is still possible that there is always a thrifty (i.e. two-stage) solution to multistage robust
set-cover having cost within $O(1) \times \opt$ (the above example just shows that the constant is at least two). If such a result were true, then in conjunction with the two-stage result from~\cite{GNR10-robust} it would also yield Theorem~\ref{thm:mult-sc}.
However, proving such an existence result seems no easier than obtaining an algorithm for the original multistage
problem. Indeed, we take the latter approach in this paper and directly give thrifty approximation algorithms for all
the multistage problems considered.


\subsection{Bad Example for Subset $k$-robust Uncertainty
  sets}\label{app:subset-k-rob} 

Consider
the following slight generalization of the $k$-robust uncertainty model: on each day $i$ we are revealed some set
$A_i\in \Omega_i$ such that $A_i$ contains the final scenario $A$; where $\Omega_i=\{S\sse U : |S\cap P_i|\le k_i\}$
consists of all subsets of $U$ having at most $k_i$ elements in a designated set $P_i\sse U$  (it can have any number of elements from $U\setminus P_i$). Again, the final scenario
$A=\cap_{i=0}^T A_i$.

Recall that we obtain the multistage $k$-robust model studied in this paper by setting $P_i=U$
(i.e. $\Omega_i={U\choose k_i}$) on all days. 
We give an example showing that thrifty algorithms perform poorly for multistage set cover under the above `subset
$k$-robust' uncertainty sets.

Consider a universe $U$ of elements partitioned into $T$ parts $\bigcup_{i=1}^T P_i$. We consider a constant number of
stages $T$. The set system consists only of singleton sets, so it suffices to talk about costs on elements. Fix a
parameter $\lambda\gg T$. For each $i\in [T]$, we have:
\begin{OneLiners} \item $|P_i|=\lambda^{i+1}$
\item Each $P_i$-element has cost $1/\lambda^i$
\item The inflation factor on day $i$ is $\lambda^{i}$
\item $k_i=1$, so $\Omega_i =\{S\sse U : |S\cap P_i|\le 1\}$
\end{OneLiners}
We first show that the optimal value is at most $T$. Note that on each day $i$, there is at most one active
$P_i$-element (i.e. $|A_i\cap P_i| \le 1$). Consider the strategy that on each day $i$ covers the unique active
$P_i$-element: the worst case cost equals $\sum_{i=1}^T \lambda^i \cdot1/\lambda^i = T$. This strategy is feasible
since $\bigcup_{i=1}^T P_i = U$.

On the other hand, we now show that any strategy that is not active on all days has cost at least $\lambda\gg T$.
Consider any strategy that is inactive on some day $j\in [T]$. One the following cases occurs:
\begin{enumerate}
\item Suppose that on every scenario sequence, {\em all} active $P_j$-elements are covered by day $j-1$. Then consider any
scenario sequence with entire $P_j$ active on day $j-1$ (this is possible since $P_j\in \Omega_{i}$ for all $i\le
j-1$). The total cost incurred on the first $j-1$ days (on this scenario sequence) is at least $|P_j|\cdot
\frac{1}{\lambda^j}=\lambda$.
 \item Suppose that there is a partial scenario sequence $\sigma$ (until day $j-1$) where an active $P_j$-element $e$ remains uncovered after day
 $j-1$. Extend $\sigma$ to a scenario sequence that has $e$ active until the end (i.e. $e$ is the unique $P_j$ element
 that remains active after day $j$). Since the algorithm is inactive on day $j$, it must be that $e$ is covered on day
 $j+1$ or later---thus the total cost along this scenario sequence is at least $\lambda^{j+1}\cdot
\frac{1}{\lambda^j}=\lambda$.
\end{enumerate}

Setting $\lambda^{T+1}=\Theta(|U|)$ it follows that thrifty strategies can be worse by a polynomial (in $|U|$) factor
for multistage robust set-cover under subset $k$-robust uncertainty sets.



}

\end{document}